\documentclass[paper=a4,UKenglish]{scrartcl}

\pdfoutput=1

\newcommand{\typesonly}[1]{}
\newcommand{\arxivonly}[1]{#1}
\newcommand{\typesorarxiv}[2]{#2}

\usepackage{microtype}

\arxivonly{
  \usepackage{authblk}

  \usepackage{amsmath,amsthm,amssymb}

  \usepackage[dvipsnames]{xcolor}
  \definecolor{darkgreen}{rgb}{0,0.45,0}
  \definecolor{darkred}{rgb}{0.75,0,0}
  \definecolor{darkblue}{rgb}{0,0,0.6}
  \usepackage[colorlinks,citecolor=darkgreen,linkcolor=darkred,urlcolor=darkblue]{hyperref}
  \usepackage{breakurl}

  \theoremstyle{plain}
  \newtheorem{theorem}{Theorem}
  \newtheorem{lemma}[theorem]{Lemma}
  \newtheorem{corollary}[theorem]{Corollary}
  \theoremstyle{definition}

  \theoremstyle{remark}
  \newtheorem*{remark}{Remark}

  \usepackage[osf]{libertine}
  \usepackage[scaled=0.96]{zi4}  
}

\newcommand{\id}{\operatorname{id}}

\usepackage{todonotes}

\newcommand{\refl}{\operatorname{refl}}
\newcommand{\pr}{\operatorname{pr}}
\newcommand{\isProp}{\operatorname{isProp}}
\newcommand{\trunc}[1]{\left\| #1 \right\|}
\newcommand{\bracket}[1]{\left( #1 \right)}
\newcommand{\UU}{\mathcal{U}}
\newcommand{\VV}{\mathcal{V}}
\newcommand{\ttt}{\star}
\newcommand{\Empty}{\mathbf{0}}
\newcommand{\unit}{\mathbf{1}}
\newcommand{\bool}{\mathbf{2}}
\newcommand{\true}{{\operatorname{t\!t}}}
\newcommand{\false}{{\operatorname{f\!f}}}
\newcommand{\inl}{{\operatorname{inl}}}
\newcommand{\inr}{{\operatorname{inr}}}
\newcommand{\idfunc}[1][]{\operatorname{id}_{#1}}

\newcommand{\N}{\mathbb{N}}
\renewcommand{\equiv}{\simeq}

\newcommand{\LEM}{\ensuremath{\operatorname{LEM}}}
\newcommand{\DNE}{\ensuremath{\operatorname{DNE}}}

\title{Parametricity, automorphisms of the universe, and excluded middle}

\typesorarxiv{
  \author[1]{Auke B. Booij}
  \author[2]{Mart{\'\i}n H. Escard{\'o}}
  \author[3]{Peter LeFanu Lumsdaine}
  \author[4]{Michael Shulman\thanks{Supported by The United States Air Force Research Laboratory under agreement number FA9550-15-1-0053.  The U.S. Government is authorized to reproduce and distribute reprints for Governmental purposes notwithstanding any copyright notation thereon.  The views and conclusions contained herein are those of the authors and should not be interpreted as necessarily representing the official policies or endorsements, either expressed or implied, of the United States Air Force Research Laboratory, the U.S. Government, or Carnegie Mellon University.}}
  \affil[1]{School of Computer Science, University of Birmingham,
    Birmingham, UK\\
    \texttt{abb538@cs.bham.ac.uk}}
  \affil[2]{School of Computer Science, University of Birmingham,
    Birmingham, UK\\
    \texttt{m.escardo@cs.bham.ac.uk}}
  \affil[3]{Mathematics Department, Stockholm University,
    Stockholm, Sweden\\
    \texttt{p.l.lumsdaine@math.su.se}}
  \affil[4]{Department of Mathematics, University of San Diego,
    San Diego, USA\\
    \texttt{shulman@sandiego.edu}}
}{
  \author[1]{Auke B. Booij}
  \author[1]{Mart{\'\i}n H. Escard{\'o}}
  \affil{School of Computer Science, University of Birmingham,
    Birmingham, UK}
  \author[2]{Peter LeFanu Lumsdaine}
  \affil{Mathematics Department, Stockholm University,
    Stockholm, Sweden}
  \author[3]{Michael Shulman\thanks{Supported by The United States Air Force Research Laboratory under agreement number FA9550-15-1-0053.  The U.S. Government is authorized to reproduce and distribute reprints for Governmental purposes notwithstanding any copyright notation thereon.  The views and conclusions contained herein are those of the authors and should not be interpreted as necessarily representing the official policies or endorsements, either expressed or implied, of the United States Air Force Research Laboratory, the U.S. Government, or Carnegie Mellon University.}}
  \affil{Department of Mathematics, University of San Diego,
    San Diego, USA}
}

\typesonly{
  \authorrunning{A.\,B.~Booij, M.\,H.~Escard{\'o}, P.\,LeF.~Lumsdaine, and M.~Shulman}

  \Copyright{Auke B. Booij, Mart{\'\i}n H. Escard{\'o}, Peter LeFanu Lumsdaine, and Michael Shulman}



\EventEditors{}
\EventNoEds{2}
\EventLongTitle{}
\EventShortTitle{}
\EventAcronym{}
\EventYear{}
\EventDate{}
\EventLocation{}
\EventLogo{}
\SeriesVolume{}
\ArticleNo{}
}

\begin{document}

\maketitle

\begin{abstract}

  It is known that one can construct non-parametric
  functions by assuming classical axioms.  Our work is a converse
  to that: we prove classical axioms in dependent type theory
  assuming specific instances of non-parametricity.
  We also address the interaction between classical axioms and the existence of automorphisms of a type universe.
  We work over intensional Martin-L\"of dependent type
  theory, and for some results assume further principles
  including function extensionality,
  propositional extensionality, propositional truncation, and
  the univalence axiom.

\arxivonly{
}
\end{abstract}

\section{Introduction: Parametricity in dependent type theory}

Broadly speaking, \emph{parametricity} statements assert that type-polymorphic
functions definable in some system must be natural in their type arguments,
in some suitable sense.
Reynolds' original theory of relational
parametricity~\cite{DBLP:conf/ifip/Reynolds83}~characterizes terms
of the polymorphically typed $\lambda$-calculus System~F.  This theory
has since been extended to richer and more expressive
type theories: to pure type systems by Bernardy,
Jansson, and Paterson~\cite{DBLP:journals/jfp/BernardyJP12}, and more
specifically to
dependent type theory by Atkey, Ghani, and Johann~\cite{DBLP:conf/popl/AtkeyGJ14}.

Most parametricity results are meta-theorems about a formal system and
make claims only about terms in the empty context.
For instance,
Reynolds' results show that the only term of System~F
with type $\forall \alpha . \alpha \to \alpha$ definable in the empty context is
the polymorphic identity function $\Lambda \alpha. \lambda (x:\alpha).x$.
Similarly, Atkey, Ghani, and Johann\ \cite[Thm.~2]{DBLP:conf/popl/AtkeyGJ14} prove that any term $f:\prod_{X : \UU} X \to X$
definable in the empty context of MLTT
must satisfy $e(f_X(a)) =f_Y(e(a))$ for all
$e : X \to Y$ and $a:X$ in their model; it follows that
$f$ acts as the identity on every type in their model, and hence no such closed term $f$ can be provably not equal to the polymorphic identity function.

Keller and Lasson showed that excluded middle is incompatible with parametricity of the universe of types (in its usual formulation)~\cite{DBLP:conf/csl/KellerL12}.
In this paper, we show, within type theory, that certain violations of
parametricity are possible if and only if certain classical principles
hold.  For example, we show that there is a function
$f : \prod_{X : \UU} X \to X$ whose value at the type $\bool$ of booleans is
different from the identity if and only if excluded middle holds
(Theorem~\ref{thm:identity-bool}, where one direction uses function
extensionality).

These are theorems \emph{of} dependent type theory, so they apply not only to closed terms but in any context, and the violations of parametricity are expressed using negations of Martin-L\"of's identity type rather than judgemental (in-)equality of terms.
Similarly, we show that excluded middle also follows from certain kinds of non-trivial automorphisms of the universe.

We work throughout in intensional Martin-L\"of type theory, with at least $\Pi$-, $\Sigma$-, identity, finite, and natural numbers types, and a universe closed under these type-formers.
For concreteness, this may be taken to be the theory of \cite{martin-lof:bibliopolis}, or of \cite[A.2]{hottbook}.
When results require further axioms---e.g.\ function extensionality, or univalence of the universe---we include these as explicit assumptions, to keep results as sharp as possible.

By the \emph{law of excluded middle}, we mean always the version from univalent foundations~\cite[3.4.1]{hottbook}, namely that $P + \neg P$
for all \emph{propositions} $P$.
Here a type is called a ``proposition'' (a ``mere proposition'' in the terminology of~\cite{hottbook}) if it
has at most one element, meaning that any two of its elements are
equal in the sense of the identity type.
Note that $\neg P$ (meaning $P\to\Empty$) is not itself a proposition unless we assume function extensionality, at least for $\Empty$-valued functions.

The \emph{propositional truncation} of a type $A$ is the universal proposition $\trunc A$ admitting a map from $A$.
We axiomatize this as in \cite[\textsection 3.7]{hottbook}, and always indicate explicitly when we are assuming it.
It is shown in~\cite{keca:toappear} that propositional truncation implies function extensionality.

When propositional truncations exist, the \emph{disjunction} of two propositions $P\vee Q$ is defined to be $\trunc{P+Q}$.
If $P$ and $Q$ are disjoint (i.e.\ $\neg(P+Q)$ holds), then $P+Q$ is already a proposition and hence equivalent to $P\vee Q$.
In particular, when we have propositional truncations, the law of excluded middle could equivalently assert that $P\vee \neg P$ for all propositions $P$.

By a \emph{logical equivalence} of types $X$ and $Y$, written $X\leftrightarrow Y$, we mean two
functions $X \to Y$ and $Y \to X$ subject to no conditions at all.

By an \emph{equivalence} of types $X$ and $Y$ we mean a function $e:X\to Y$ that has both a left and a right inverse, i.e.\ functions $s,r:Y\to X$ with $e(s(y))=y$ for all $y:Y$ and $r(e(x))=x$ for all $x:X$.
This notion of equivalence is logically equivalent to having a \emph{single} two-sided inverse, which is all that we will need in this paper. But the notion of equivalence is better-behaved in univalent foundations (see~\cite[Chapter 4]{hottbook}); the reason is that the type expressing ``being an equivalence'' is a proposition, in the presence of function extensionality, whereas the type expressing ``having a two-sided inverse'' may in general have more than one inhabitant, in particular affecting the consistency of the univalence axiom.

\section{Classical axioms from non-parametricity}
\label{parametricity:section}

In this section, we give a number of ways in which classical axioms can be derived
from specific violations of parametricity.

\subsection{Polymorphic endomaps}

Say that a function $f:\prod_{X : \UU} X \to X$ is \emph{natural under
  equivalence} if for any two types $X$ and $Y$ and any equivalence
$e:X \to Y$, we have $e(f_X(x)) = f_Y(e(x))$ for any $x:X$, where we have
written $f_X$ as a shorthand for $f(X)$ and used the equality sign $=$ to denote identity types.
\begin{theorem}
\label{thm:identity-bool}
If there is a function $f:\prod_{X : \UU} X \to X$ such that $f_\bool$
is not pointwise equal to the identity (i.e.\ $\neg \prod_{x:\bool} f_\bool(x)=x$)
and $f$ is natural under equivalence, then the law
of excluded middle holds.
Assuming function extensionality, the converse also holds.
\end{theorem}

\begin{proof}
  First we derive excluded middle from~$f$.
  To begin, note that if
  \arxivonly{\linebreak[4]}
  $\neg \prod_{x:\bool} f_\bool(x)=x$, then we cannot have both $f_\bool(\true)=\true$ and $f_\bool(\false)=\false$, since then we could prove $\prod_{x:\bool} f_\bool(x)=x$ by case analysis on $x$.
  But then by case analysis on $f_\bool(\true)$ and $f_\bool(\false)$, we must have $(f_\bool(\true)=\false)+(f_\bool(\false)=\true)$.
  Without loss of generality, suppose $f_\bool(\true)=\false$.

  Now let $P$ be an arbitrary
  proposition. We do case analysis on
  $f_{P+\unit}(\inr(\ttt)) : P+\unit$.
  \begin{enumerate}
  \item If it is of the form $\inl(p)$ with $p:P$, we conclude
    immediately that $P$ holds.
  \item If it is of the form $\inr(\ttt)$, then $P$ cannot hold, for if
    we had $p:P$, then the map $e : \bool \to P + \unit$ defined by $e(\false)=\inl(p)$ and $e(\true)=\inr(\ttt)$ would be an equivalence, and hence $e(f_\bool(x)) = f_{P+\unit}(e(x))$ for all $x:\bool$ and so
  $ \inl(p) = e(\false) = e (f_\bool(\true)) = f_{P+\unit} (e(\true)) = f_{P+\unit} (\inr(\ttt)) = \inr(\ttt)$, which is a contradiction.
  \end{enumerate}
Therefore $P$ or not $P$.

For the converse, \cite[Exercise~6.9]{hottbook}, suppose
excluded middle holds, let $X:\UU$ and $x:X$, and consider the type
$\sum_{x':X} (x'\neq x)$, where $a\neq b$ means $\neg(a=b)$.  By excluded middle, this is either
contractible or not.  (A type $Y$ is \emph{contractible} if
$\sum_{y:Y} \prod_{y':Y} (y=y')$.  Assuming
function extensionality, this is a proposition.) If it is contractible, define $f_X(x)$ to
be the center of contraction (the point $y$ in the definition of
contractibility); otherwise define $f_X(x)=x$.
\end{proof}

\begin{remark} \leavevmode
  \begin{enumerate}
  \item If we assume univalence, any $f:\prod_{X : \UU} X \to X$ is automatically
    natural under equivalence, so that assumption can be dispensed with.
    And, of course, if function extensionality holds (which follows
    from univalence) then the hypothesis $\neg \prod_{x:\bool} f_\bool(x)=x$
    is equivalent to $f_\bool \neq \lambda (x:\bool). x$.
  \item We do not know whether the converse direction of
    Theorem~\ref{thm:identity-bool} is provable without function
    extensionality.
  \end{enumerate}
\end{remark}

The preceding proof can be generalized as follows.
We say that a point $x:X$ is \emph{isolated} if
the type $x=y$ is decidable for all $y:Y$, i.e.\ if we have $\prod_{y:X} (x=y)+(x\neq y)$.

\begin{lemma}\label{lem:isolated}
  A point $x:X$ is isolated if and only if $X$ is equivalent to
  $Y+\unit$, for some type~$Y$, by a map that sends $x$ to
  $\inr(\ttt)$.
\end{lemma}
\begin{proof}
  Since $\inr(\ttt)$ is isolated, such an equivalence certainly
  implies that $x$ is isolated.  Conversely, from
  $\prod_{y:X} (x=y)+(x\neq y)$ we can construct a function
  $d:X\to\bool$ such that $d(y)=\true$ if and only if $x=y$ and $d(y)=\false$, and if and only if
  $x\neq y$.  Let $Y$ be $\sum_{y:X} (d(y)=\false)$; it is straightforward
  to show $X\equiv Y+\unit$.

  If we had function extensionality (for $\Empty$-valued functions),
  we could dispense with $d$ and define $Y = \sum_{y:X} (x\neq y)$,
  since then $x\neq y$ would be a proposition.  In general we use
  $d(y)=\false$ as it is always a proposition (since $\bool$ has
  decidable equality, hence its identity types are propositions by
  Hedberg's theorem); this is necessary to show that the composite
  $Y+\unit \to X \to Y+\unit$ acts as the identity on $Y$.
\end{proof}

%
\begin{theorem}
\label{thm:identity-isolated}
  If there is a function $f:\prod_{X : \UU} X \to X$ such that
  $f_X(x)\neq x$ for some isolated point $x:X$ and $f$ is natural
  under equivalence, then the law of excluded middle holds.
  Assuming function extensionality, the converse also holds.
\end{theorem}

\begin{proof}
  To derive excluded middle from~$f$, let $Y$ and $X\equiv Y+\unit$ be as in Lemma~\ref{lem:isolated}, and let $P$ be an arbitrary
  proposition. We do case analysis on
  $f_{P\times Y + \unit}(\inr(\ttt)) : P \times Y + \unit$. 
  \begin{enumerate}
  \item If it is of the form $\inl((p,y))$ with $p:P$, we conclude
    immediately that $P$ holds.
  \item If it is of the form $\inr(\ttt)$, then $P$ cannot hold, for
    if we had $p:P$, then the map $e : X \to P\times Y + \unit$
    defined by $e(x)=\inr(\ttt)$ (where $x$ is the isolated point) and
    $e(y)=\inl((p,y))$ for $y \neq x$ would be an equivalence, and
    hence $e(f_X(x)) = f_{P\times Y+\unit}(e(x))$, and so
    $ \inl((p,f_X(x))) = e (f_X(x)) = f_{P\times Y+\unit} (e(x)) =
    f_{P\times Y+\unit} (\inr(\ttt)) =\inr(\ttt)$, which is a contradiction.
  \end{enumerate}
Therefore either $P$ or not $P$ holds.
The converse is proven exactly as in Theorem~\ref{thm:identity-bool}.
\end{proof}

Finally, if our type theory includes propositional truncations, we can dispense with
isolatedness.

\begin{theorem}
\label{thm:identity-proptrunc}
  In a type theory with propositional truncations, there is an
  \arxivonly{\linebreak[4]}
  equivalence-natural function $f : \prod_{X:\UU} X\to X$ and a type
  $X : \UU$ with a point $x:X$ such that $f_X(x)\neq x$ if and only if excluded middle holds.
\end{theorem}
\begin{proof}
  For the ``if'' direction, note that propositional truncation
  implies function extensionality~\cite{keca:toappear}, so the converse direction of
  Theorem~\ref{thm:identity-bool} applies.  For the ``only if'' direction,
  assume that we are given $f : \prod_{X:\UU} X\to X$, a type $X:\UU$
  and a point $x:X$ with $f_X(x) \ne x$. Let $P$ be any proposition,
  and define
\begin{gather*}
Z  = \sum_{y:X} \trunc{x=y} \vee P, \qquad
z =  (x,|\inl(|\refl_x|)|) : Z, \qquad
y =  \pr_1(f_Z(z)) : X.
\end{gather*}
Recall that $A\vee B$ denotes the truncated disjunction $\trunc{A+B}$.
This binds more tightly than $\Sigma$, so $Z = \sum_{y:X} (\trunc{x=y} \vee P)$.
We write $|a|:\trunc{A}$ for the witness induced by a point $a:A$.

Now the second projection $\pr_2(f_Z(z))$ tells us that
$\trunc{x=y} \vee P$.  However, if $P$ holds, then $\pr_1:Z \to X$ is an
equivalence that maps $z$ to $x$. Thus $f_Z(z) \neq z$ and
hence $x\neq y$.
In other words, $P\to (x\neq y)$, hence $(x=y) \to \neg P$ and so also $\trunc{x=y}\to \neg P$.
But since $\trunc{x=y} \vee P$, we have $\neg P \vee P$, which (in the presence of function extensionality) is equivalent to excluded middle.
\end{proof}

\begin{remark}\leavevmode
  \begin{enumerate}
  \item If $x:X$ happens to be isolated, then the type $Z$ defined in
    the proof of Theorem~\ref{thm:identity-proptrunc} is equivalent to
    the type $P\times Y+\unit$ used in the proof of
    Theorem~\ref{thm:identity-isolated}.
  \item Since propositional truncation implies function extensionality~\cite{keca:toappear},
    it makes excluded middle into a proposition.  Thus, the existence hypothesis of
    Theorem~\ref{thm:identity-proptrunc} can be truncated or untruncated without change of meaning.
  \item  The hypothesis can also be formulated as
    ``there is a type $X$ such that $f_X$ is apart from the
    identity of $X$'', where two functions $g,h:A \to B$ of types $A$
    and $B$ are \emph{apart} if there is $a:A$ with $g(a) \ne
    h(a)$. We don't know whether it is possible to derive
    excluded middle from the weaker assumption that $f_X$ is simply
    \emph{unequal} to the identity function of $X$, or even that $f$
    is unequal to the polymorphic identity function.
  \end{enumerate}
\end{remark}

The above can be applied to obtain classical axioms from
other kinds of violations of parametricity.  As a simple example,
consider $f:\prod_{X : \UU} (X \to X) \to (X \to X)$.  Parametric
elements of this type are Church numerals.
Given $f$, we can define a polymorphic endomap
$g:\prod_{X : \UU} X \to X$ by $g_X=f_X(\idfunc[X])$, where
$\idfunc[X]$ is the identity function.  If $f$ is natural under
equivalence, then so is $g$, and hence the assumption that
$f_\bool(\idfunc[\bool])$ is \emph{not} the identity function gives
excluded middle, assuming function extensionality.

\subsection{Maps of the universe into the booleans}
\label{sec:u-to-bool}

A function $f : \UU \to \bool$ is \emph{invariant under equivalence}, or
\emph{extensional}, if we have $f(X)=f(Y)$ for any two equivalent types $X$
and~$Y$. We say that it is \emph{strongly non-constant} if we have
$X,Y:\UU$ with $f(X) \neq f(Y)$.
Assuming function extensionality, Escard{\'o} and
Streicher~\cite[Thm.~2.2]{DBLP:journals/apal/EscardoS16} showed that if $f : \UU \to
\bool$ is extensional and strongly non-constant, then the weak limited
principle of omniscience holds (any function $\N \to \bool$ is
constant or not).  Alex Simpson strengthened this as follows (also reported
in~\cite[Thm.~2.8]{DBLP:journals/apal/EscardoS16}):

\begin{theorem}[Simpson]\label{thm:simpson}
Assuming function extensionality for $\Empty$-valued functions,
there is an extensional, strongly non-constant function $f:\UU \to
\bool$ if and only if weak excluded middle holds (meaning that
$\neg A + \neg\neg A$ for all $A : \UU$).
\end{theorem}
\begin{proof}
  In one direction, suppose weak excluded middle, and define $f:\UU \to
  \bool$ by $f(A)=\false$ if $\neg A$ and $f(A)=\true$ if $\neg\neg A$.
  Then $f(\Empty)=\false$ and $f(\unit)=\true$, so $f$ is strongly
  non-constant.  Extensionality follows from the observation that if
  $A\equiv B$ then $\neg A \leftrightarrow \neg B$ and
  $\neg \neg A \leftrightarrow \neg \neg B$.

  In the other direction, suppose $f:\UU \to \bool$ is extensional,
  and strongly non-constant witnessed by types $X,Y:\UU$ with
  $f(X) \neq f(Y)$.  Suppose without loss of generality that
  $f(X)=\true$ and $f(Y)=\false$.  For any $A:\UU$, define
  $Z = \neg A \times X + \neg\neg A \times Y$.  If $A$, then
  $\neg A\equiv\Empty$ and $\neg\neg A \equiv \unit$
  (using function extensionality), so $Z\equiv Y$
  and $f(Z)=\false$.  Similarly, if $\neg A$, then $Z\equiv X$ and so
  $f(Z)=\true$.  On the other hand, $f(Z)$ must be either $\true$ or
  $\false$ and not both.  If it is $\true$, then it is not $\false$,
  and so $\neg A$; while if it is $\false$, then it is not $\true$,
  and so $\neg\neg A$.
\end{proof}

In
Theorem~\ref{polymorphic:into:booleans} below we reuse Simpson's
argument to establish a similar conclusion for polymorphic functions
into the booleans.

\subsection{Polymorphic maps into the booleans}
\label{sec:ptd-u-to-bool}

A function $f:\prod_{X : \UU} X \to \bool$ is \emph{invariant under
equivalence} if we have $f_Y(e(x)) = f_X(x)$ for any equivalence
$e:X\to Y$ and point $x:X$.
Such a function ``violates parametricity'' if it is non-constant.
Equivalence invariance means that some such violations are literally impossible: for instance, there cannot be a type $X$ with points $x,y:X$ such that $f_X(x) \neq f_X(y)$ if there is an automorphism of $X$ that maps $x$ to $y$.

A violation of constancy \emph{across} types, rather than at a
specific type, is equivalent to weak excluded middle.
\begin{theorem} \label{polymorphic:into:booleans}
  Assuming function extensionality for $\Empty$-valued functions,
  weak excluded middle holds if and only if there is an
  $f:\prod_{X:\UU} X \to \bool$ that is invariant under equivalence,
  together with $X,Y:\UU$ with isolated points $x:X$ and $y:Y$ such
  that $f_X(x) \neq f_Y(y)$.
\end{theorem}
\begin{proof}
  Assuming weak excluded middle, to show the existence of such an $f$,
  let $X:\UU$ and $x:X$.
  Then use weak excluded middle to decide $\neg(\sum_{x':X}x\neq
  x')+\neg\neg(\sum_{x':X}x\neq x')$.
  In the left case, expressing that there are no other elements in $X$
  than $x$, define $f_X(x)=\false$, and in the right case define
  $f_X(x)=\true$.
  So, for example, $f_\unit(\ttt)=\false$ and $f_\bool(\true)=\true$,
  showing that we constructed a non-constant $f$ as required.

  For the other direction, without loss of generality, $f_X(x)=\true$
  and $f_Y(y)=\false$.
  By assumption, $X$ is equivalent to $\unit+X'$ via an equivalence that sends $x$ to
  $\inl(\ttt)$, and similarly $Y$ is equivalent to $\unit+Y'$ via an
  equivalence that sends $y$ to $\inl(\ttt)$.
  Let $A:\UU$ and define
  \begin{eqnarray*}
    Z & = & (\unit + \neg A \times X') \times (\unit + \neg\neg A
            \times Y'), \\
    z & = & (\inl(\ttt),\inl(\ttt)).
  \end{eqnarray*}
  By the invariance under equivalence of $f$,
  \begin{enumerate}
  \item if $\neg A$ then $Z\equiv X$ via an equivalence
    that sends $z$ to $x$, thus $f_Z(z)=\true$,
  \item if $A$ then $Z\equiv Y$ via an equivalence that sends $z$ to
    $y$, thus $f_Z(z)=\false$.
  \end{enumerate}
  The contrapositives of these two implications are respectively
  \begin{eqnarray*}
    f_Z(z)=\false & \to & \neg\neg A,\\
    f_Z(z)=\true & \to & \neg A.
  \end{eqnarray*}
  Hence we can decide $\neg A$ by case analysis on the value of $f_Z(z)$.
\end{proof}

Provided our type theory includes propositional truncations, we can
dispense with isolatedness as in Theorem~\ref{thm:identity-proptrunc},
assuming the types $x=x$ and $y=y$ are propositions.
\begin{theorem}\label{thm:ptd-u-to-bool}
  In a type theory with propositional truncations, weak excluded
  middle holds if and only if there is an
  $f:\prod_{X:\UU} X \to \bool$ that is invariant under equivalence,
  together with $X,Y:\UU$ with $x:X$ and $y:Y$ such that
  $f_X(x) \neq f_Y(y)$, where the types $x=x$ and $y=y$ are
  propositions.
\end{theorem}
\begin{proof}
  Assuming weak excluded middle, the existence of such an $f$ is shown
  as in the proof of Theorem~\ref{polymorphic:into:booleans}.

  For the other direction, without loss of generality, $f_X(x)=\true$
  and $f_Y(y)=\false$.
  Note that since $x=x$ and $y=y$ are propositions, so are $x=x'$ and $y=y'$ for any $x':X$ and $y':Y$, since as soon as they have a point they are equivalent to $x=x$ and $y=y$ respectively.
  Let $A:\UU$ and define
  \begin{eqnarray*}
    Z & = & \left(\sum_{x':X}\bracket{x=x'}\vee\neg A\right) \times
            \left(\sum_{y':Y}\bracket{y=y'}\vee\neg\neg A\right), \\
    z & = & ((x,|\inl(\refl)|),(y,|\inl(\refl)|)).
  \end{eqnarray*}
  By invariance under equivalence of $f$, we have the following.
  \begin{enumerate}
  \item If $\neg A$ then $Z\equiv X$ via an equivalence that sends $z$
    to $x$, thus $f_Z(z)=\true$.  This works because the left factor
    of $Z$ becomes equivalent to $X$, and the right factor equivalent
    to $\unit$ by the assumptions that $y=y$ is a proposition and
    $\neg A$.
  \item Similarly, if $A$ then $Z\equiv Y$ via an equivalence that
    sends $z$ to $y$, thus $f_Z(z)=\false$, now using the fact that
    $x=x$ is a proposition.
  \end{enumerate}
  The contrapositives of these two implications are respectively
  \begin{eqnarray*}
    f_Z(z)=\false & \to & \neg\neg A, \\
    f_Z(z)=\true & \to & \neg A.
  \end{eqnarray*}
  Hence we can decide $\neg A$ by case analysis on the value of $f_Z(z)$.
\end{proof}

\begin{remark}
  In a type theory with pushouts, the assumptions that $x=x$ and $y=y$
  are propositions can be removed by using the join $(x=x') * \neg A$
  instead of the disjunction $\bracket{x=x'}\vee \neg A$ in the left
  factor of $Z$, and similarly for the right factor of $Z$.  (The
  \emph{join} $B*C$ of types $B$ and $C$ is the pushout of $B$ and $C$
  under $B\times C$.)  This works since joining with an empty type is
  the identity, while joining with a contractible type gives a
  contractible result; see Theorem~\ref{thm:ptd-decomp} below for
  details.
  Indeed, the join of two propositions is their disjunction, by~\cite[Lemma 2.4]{rijke:join}; but the version using joins does not quite subsume the one using disjunctions, since if joins are not already assumed to exist, we do not know how to show that the disjunction of two propositions is their join.
\end{remark}

\subsection{Decompositions of the universe}
\label{sec:decomp}

Theorem~\ref{thm:simpson} can be
interpreted as saying that the universe $\UU$ cannot be decomposed
into two disjoint inhabited parts without weak excluded middle.  In
fact, disjointness of the parts is not necessary.  All that is needed
is that both parts be proper, i.e.\ not the whole of $\UU$:

\begin{theorem}
  In a type theory with propositional truncation and function
  extensionality for $\Empty$-valued functions, suppose we have
  equivalence-invariant $P,Q:\UU\to\UU$ such that for all $Z:\UU$ we
  have $P(Z) \vee Q(Z)$, and that we have types $X$ and $Y$ such that
  $\neg P(X)$ and $\neg Q(Y)$.  Then weak excluded middle holds.
\end{theorem}
\begin{proof}
  For any $A:\UU$, let $Z = \neg A \times X + \neg\neg A \times Y$ as
  in Simpson's proof.  If $A$, then $Z\equiv Y$, and so $\neg Q(Z)$;
  thus $Q(Z) \to \neg A$.  But if $\neg A$, then $Z\equiv X$, and so
  $\neg P(Z)$; thus $P(Z) \to \neg\neg A$.  Hence the assumed $P(Z)
  \vee Q(Z)$ implies $\neg A \vee \neg \neg A$, which is equivalent to
  $\neg A + \neg\neg A$ since $\neg A$ and $\neg\neg A$ are (by
  function extensionality) disjoint propositions.
\end{proof}

The proof of Theorem~\ref{thm:ptd-u-to-bool} can be similarly adapted.

\begin{theorem}\label{thm:ptd-decomp}
  In a type theory with propositional truncation and $\Empty$-valued
  function extensionality, suppose we have $P,Q:\prod_{X:\UU} X \to
  \UU$ that are invariant under equivalence, i.e.\ if $X\equiv Y$ by an
  equivalence sending $x:X$ to $y:Y$, then $P_X(x) \equiv P_Y(y)$, and
  likewise for $Q$.  Suppose also that for all $Z:\UU$ and $z:Z$ we
  have $P_Z(z) \vee Q_Z(z)$, and types $X,Y$ with points $x:X$ and
  $y:Y$ such that $\neg P_X(x)$ and $\neg Q_Y(y)$.  Finally, suppose
  either that our type theory has pushouts or that the types $x=x$ and
  $y=y$ are propositions.  Then weak excluded middle holds.
\end{theorem}
\begin{proof}
  For variety in contrast to Theorem~\ref{thm:ptd-u-to-bool},
  suppose we have pushouts; we leave the other case to the reader.
  Let $A:\UU$ and define
  \begin{eqnarray*}
    Z & = & \left(\sum_{x':X} (x=x') * \neg A\right) \times
            \left(\sum_{y':Y} (y=y')* \neg\neg A\right), \\
    z & = & ((x,\inl(\refl)),(y,\inl(\refl))).
  \end{eqnarray*}
  Then if $A$, $\neg A\equiv \Empty$, so $(x=x')*\neg A \equiv (x=x')$,
  and thus the first factor of $Z$ is equivalent to $\sum_{x':X}
  (x=x')$, which is a ``singleton'' or ``based path space'' and hence
  equivalent to $\unit$.  On the other hand (still assuming $A$),
  $\neg\neg A\equiv \unit$, so $(y=y')*\neg\neg A \equiv \unit$, and
  thus the right factor of $Z$ is equivalent to $\sum_{y':Y}\unit$ and
  hence to $Y$.  Thus, $A$ implies $Z\equiv Y$, and it is easy to check
  that this equivalence sends $z$ to $y$.  Hence $A \to \neg Q_Z(z)$,
  and so $Q_Z(z) \to \neg A$.  A dual argument shows that $\neg A \to
  \neg P_Z(z)$ and thus $P_Z(z) \to \neg\neg A$, so the assumption
  $P_Z(z) \vee Q_Z(z)$ gives weak excluded middle.
\end{proof}

Since a function $\prod_{X:\UU} X \to B$, for any fixed $B$, is the
same as a function $\left(\sum_{X:\UU} X\right) \to B$, we can
interpret Theorem~\ref{thm:ptd-decomp} as saying that the universe $\sum_{X : \UU} X$ of
\emph{pointed types} also cannot be decomposed into two proper parts
without weak excluded middle.

The results discussed so far illustrate that different violations of
parametricity have different proof-theoretic strength: some violations
are impossible, while others imply varying amounts of excluded middle.

\section{Classical axioms from automorphisms of the universe}
\label{classical:automorphisms}

There have been attempts to apply parametricity to show that the only
automorphism of a universe of types is the identity.  Nicolai Kraus
observed in the HoTT mailing list~\cite{automorphisms:kraus} that,
assuming univalence, automorphisms of a universe $\UU$ living in a
universe $\mathcal{V}$ correspond to elements of the loop space%
\footnote{The \emph{loop space} $\Omega(X,x)$ of a type $X$ at a point $x:X$ is the identity type $x=x$; see~\cite[\textsection 2.1]{hottbook}.}
$\Omega(\mathcal{V},\UU)$, while elements of the higher loop space $\Omega^2(\mathcal{V},\UU)$
correspond to ``polymorphic automorphisms'' $\prod_{X:\UU} X\equiv X$, which are at least as strong as polymorphic endomaps.
In particular, nontrivial elements of $\Omega^2(\mathcal{V},\UU)$ imply violations of
parametricity for $\prod_{X : \UU} X \to X$.  This suggests that
parametricity may play a role in automorphisms of the universe.

We are not aware of a proof that parametricity implies that the only
automorphism of the universe is the identity. However, in the spirit
of the above development, we can show that automorphisms with specific
properties imply excluded middle.  First, however, we observe that if we do have excluded middle then we can construct various nontrivial automorphisms of the universe.

\subsection{Automorphisms from excluded middle}
\label{sec:autom-that-do}

The simplest automorphism of the universe is defined as follows.
By \emph{propositional extensionality} we mean that any two logically equivalent propositions are equal.
(This follows from propositional univalence, i.e.\ univalence asserted only for propositions.
The converse holds at least assuming function extensionality; we do not know whether this assumption is necessary.)

\begin{theorem}\label{thm:swap10}
  Assuming excluded middle, function extensionality, and propositional extensionality, there is an automorphism $f:\UU\equiv \UU$ such that $f(\unit)\equiv\Empty$.
\end{theorem}
\begin{proof}
Given a type
$X$, we use excluded middle to decide if it is a proposition (this works because under function extensionality, being a proposition is itself a proposition).  If it
is, we define $f(X)=\neg X$, and otherwise we define $f(X)=X$.
Assuming propositional extensionality and excluded middle,
we have $\neg\neg X = X$ for any proposition; thus $f(f(X))=X$ whether $X$ is a proposition or not, and hence $f$ is a self-inverse equivalence.
\end{proof}


We can try to construct other automorphisms of the universe by permuting some other subclass of types.
For instance, if we have propositional truncation, then given any two non-equivalent types $A$ and $B$, excluded middle implies that for any type $X$ we have $\trunc{X=A} + \trunc{X=B} + (X\neq A \wedge X\neq B)$, so that the universe $\UU$ decomposes as a sum $\UU_A + \UU_B + \UU_{\neq A,B}$, where
\begin{gather*} 
  \UU_A = \sum_{X:\UU} \trunc{X=A}, \qquad
  \UU_B = \sum_{X:\UU} \trunc{X=B}, \qquad
  \UU_{\neq(A,B)} = \sum_{X:\UU} (X\neq A \wedge X\neq B).
\end{gather*}
(This requires function extensionality for $X\neq A$ and $X\neq B$ to be propositions, but not univalence.)
Thus, if $\UU_A \equiv \UU_B$ we can switch those two summands to produce an automorphism of $\UU$:

\begin{theorem}\label{thm:uuequiv-aut}
  Assuming function extensionality and excluded middle, if $A\not\equiv B$ and $\UU_A \equiv \UU_B$, then there is an automorphism $f:\UU\equiv \UU$ such that $\trunc{f(A)=B}$, hence $f\neq\id$.
\end{theorem}
\begin{proof}
  We use the above decomposition and the given equivalence $\UU_A \equiv \UU_B$ to produce $f$.
  And since $f$ maps $\UU_A$ to $\UU_B$, by definition of $\UU_B$ we have $\trunc{f(A)=B}$.
\end{proof}

This leads to the question, when can we have $\UU_A \equiv \UU_B$ but $A \not\equiv B$?
Theorem~\ref{thm:swap10} is the simplest example of this: assuming propositional extensionality, both $\UU_{\Empty}$ and $\UU_{\unit}$ are contractible, hence equivalent to $\unit$.
More generally, let us call a type $X$ \emph{rigid} if $\UU_X$ is contractible; then we have:

\begin{theorem}
  Assuming function extensionality and excluded middle, if $A$ and $B$ are rigid types with $A\not\equiv B$, then there is an automorphism $f:\UU\equiv \UU$ such that $f(A)\equiv B$.
\end{theorem}
\begin{proof}
  This follows from Theorem \ref{thm:uuequiv-aut}.
  In the rigid case we get the stronger conclusion that $f(A)\equiv B$, since $\UU_B$ is contractible.
\end{proof}

More generally, under excluded middle any permutation of the rigid types yields an automorphism of the universe.

If we assume UIP, then \emph{every} type is rigid, so that with UIP and excluded middle there are plenty of automorphisms of the universe.
If we instead assume univalence --- as we will do for the rest of this subsection --- most types are not rigid.
For instance, any type with two distinct isolated points, such as $\mathbb{N}$, is not rigid, since we can swap the isolated points to give a nontrivial automorphism and hence a nontrivial equality in $\UU_X$.
In particular, if excluded middle holds and $X$ is a \emph{set} (i.e.\ its identity types are all propositions), then all points of $X$ are isolated.
Thus, with excluded middle and univalence, no set with more than one element (i.e.\ with points $x,y:X$ such that $x\neq y$) is rigid.

However, there exist types that are \emph{connected} (i.e.\ $\trunc{X}$ and  $\prod_{x,y:X} \trunc{x=y}$), but that are not trivial; indeed, as remarked above, $\UU_A$ is such a type.
Moreover, if we also assume higher inductive types, then from any group $G$ that is a set we can construct a connected type $BG$ such that $\Omega(BG) \equiv G$ \cite[\textsection 3.2]{lf:emspaces}.

This leads us to ask, when is $BG$ rigid for a set-group $G$?
Since $BG$ is a 1-type (i.e.\ its identity types are all sets), $\UU_{BG}$ is a 2-type (i.e.\ its identity types are all 1-types).
Hence it is contractible as soon as its loop space is connected and its double loop space is contractible.
In general, the connected components of $\Omega(\UU_{BG})$ are the \emph{outer automorphisms} of $G$ (equivalence classes of automorphisms of $G$ modulo conjugation), while $\Omega^2(\UU_{BG})$ is the \emph{center} of $G$ (the subgroup of elements that commute with everything).
A group with trivial outer automorphism group and trivial center is sometimes known as a \emph{complete} group (though there is no apparent relation to any topological notion of completeness), and there are plenty of examples.

For instance, the symmetric group $S_n$ is complete in this sense except when $n=2$ or $6$.
Thus, $BS_n$ is rigid for $n\notin \{2,6\}$.
(Note also that $BS_n$ can be constructed without higher inductive types --- but with univalence --- as $\UU_{[n]}$, where $[n]$ is a finite $n$-element type, although of course this type only lives in a larger universe $\mathcal{V}$.)
In particular, assuming univalence and excluded middle, there are countably infinitely many rigid types, and hence \emph{uncountably} many nontrivial automorphisms of $\UU$ (one induced by every permutation of the types $BS_n$ for $n\notin \{2,6\}$).

This does not exhaust the potential automorphisms of $\UU$.
For instance, we have:

\begin{theorem}\label{thm:swap-rigid-prod}
  Let $X$ be an $n$-type for some $n\ge -1$, and let $A$ and $B$ be $n$-connected rigid types such that $X\times A \not\equiv X\times B$.
  Then assuming univalence and excluded middle, there is an automorphism $f:\UU\equiv\UU$ such that $\trunc{f(X\times A) = (X\times B)}$.
\end{theorem}
\begin{proof}
  We will show that $\UU_{X\times A} \simeq \UU_{X\times B}$, by showing that both are equivalent to $\UU_X$.
  It suffices to consider $A$.
  We have $(Z\mapsto Z\times A) : \UU_X \to \UU_{X\times A}$, and since both types are connected it suffices to show that it induces an equivalence of loop spaces $\Omega\UU_X \to \Omega\UU_{X\times A}$, or equivalently that the induced map $L:(X\equiv X) \to (X\times A \equiv X\times A)$ is an equivalence.
  Since $A$ is $n$-connected for $n\ge -1$, we have $\trunc A$; so since being an equivalence is a proposition we may assume given $a_0:A$.

  We claim that for all $a:A$, $x:X$, and $f:X\times A \to X\times A$ we have
  \begin{equation}
    \pr_1(f(x,a)) = \pr_1(f(x,a_0)).\label{eq:rigid-prod-star}
  \end{equation}
  Since this goal is an equality in the $n$-type $X$, it is an $(n-1)$-type.
  And since $A$ is $n$-connected, the map $a_0:\unit\to A$ is $(n-1)$-connected by~\cite[Lemma 7.5.11]{hottbook}.
  Thus, by~\cite[Lemma 7.5.7]{hottbook}, it suffices to assume that $a=a_0$, in which case~\eqref{eq:rigid-prod-star} is clear.

  It follows from~\eqref{eq:rigid-prod-star} that if we define $M:(X\times A \to X\times A)\to (X\to X)$ by $M(f)(x) = \pr_1(f(x,a_0))$, then $M$ preserves composition and identities.
  Thus it preserves equivalences, inducing a map $(X\times A \equiv X\times A)\to (X\equiv X)$.
  We easily have $M\circ L = \id$, so to prove $L\circ M = \id$ it suffices to show that $M$ is left-cancellable, i.e.\ that $(M f = M g) \to (f = g)$.
  Since $M$ preserves composition, for this it suffices to show that if $M f = \id$ then $f = \id$.
  But if $M f = \id$, then by~\eqref{eq:rigid-prod-star} we have $\pr_1(f(x,a)) = x$ for all $a:A$.
  Thus $f(x,a) = (x,g_x(a))$, where $g_x:A\equiv A$ for each $x:X$.
  But $A$ is rigid, so each $g_x = \id$, hence $f=\id$.
\end{proof}

For instance, we could take $n=0$ and $X=\bool$, so that $X\times A \equiv A+A$.
Thus if $A$ and $B$ are any connected rigid types, an automorphism of $\UU$ can swap $A+A$ with $B+B$.

There might also be rigid types that are not of the form $BG$, or types $A,B$ not built out of rigid ones but such that $A\not\equiv B$ and $\UU_A\equiv \UU_B$.
But now we will leave such questions and turn to the converse: when does an automorphism of $\UU$ imply excluded middle?

\subsection{Excluded middle from automorphisms}
\label{sec:excluded-middle-from}

In fact, without function extensionality, we can only derive a slightly weaker form of excluded middle from a nontrivial automorphism of the universe.
As defined in the introduction, the \emph{law of excluded middle} (\LEM) is
\[
\prod_{P:\UU} \isProp(P) \to P + \neg P.
\]
We will instead derive the law of \emph{double-negation elimination} (\DNE), which is
\[
\prod_{P:\UU} \isProp(P) \to \neg\neg P \to P.
\]
Notice that if $\Empty$-valued function extensionality holds, then
$\neg P$ is a proposition (even if $P$ is not a proposition) and hence,
if $P$ is a proposition, $P+\neg P$ is a proposition equivalent to
$P\vee \neg P$.
In first-order or higher-order logic, the corresponding schemas or
axioms of excluded middle and double-negation elimination are
equivalent, but, in type theory, one direction seems to require some
amount of function extensionality:
\begin{lemma} \label{lem:is:dne} \leavevmode
  \begin{enumerate}
  \item \label{lem:gives:dne} \LEM\ implies \DNE.
  \item \label{dne:gives:lem} \DNE\ implies \LEM\, assuming $\Empty$-valued function extensionality.
  \end{enumerate}
\end{lemma}
\begin{proof}
  (\ref{lem:gives:dne}):~Assume \LEM\ and let $P:\UU$ with
  $\isProp(P)$ and assume $\neg\neg P$.  By excluded middle, either
  $P$ or $\neg P$. In the first case we are done, and the second
  contradicts $\neg \neg P$.
  (\ref{dne:gives:lem}):~Assume \DNE\ and let $P:\UU$ with
  $\isProp(P)$. By $\Empty$-valued function extensionality,
  $P + \neg P$ is a proposition, and hence \DNE\ gives $P + \neg P$,
  because we always have $\neg \neg (P + \neg P)$.
\end{proof}

\begin{lemma}
\label{dne:dne-negation}
\DNE\ holds if and only if every proposition is logically equivalent
to the negation of some type.
\end{lemma}
\begin{proof}
  ($\Rightarrow$):~\DNE\ gives that any proposition $P$ is logically
  equivalent to the negation of the type $\neg P$.
  ($\Leftarrow$):~For any two types $A$ and $B$, we have that
  $A \to B$ implies $\neg B \to \neg A$. Hence $A \to B$ also gives
  $\neg \neg A \to \neg \neg B$. And, because $X \to \neg \neg X$ for
  any type $X$, we have $\neg \neg \neg X \to \neg X$. Therefore, if
  $P$ is logically equivalent to the negation of $X$, we have the
  chain of implications
  $\neg \neg P \to \neg \neg \neg X \to \neg X \to P$.
\end{proof}

Our first automorphism of the universe constructed from excluded
middle swapped the empty type with the unit type.  We now show that conversely, any
such automorphism implies \DNE\, and hence, assuming $\Empty$-valued
function extensionality, also~\LEM.
In fact, not even an \emph{embedding} of $\UU$ into itself that maps
the unit type to the empty type is possible without classical axioms:
\begin{theorem} \label{left:cancellable} Assuming propositional
  extensionality, if there is a left-cancellable map $f: \UU \to \UU$
  with $f(\unit)=\Empty$, then \DNE\ holds.
\end{theorem}
\begin{proof}
  For an arbitrary proposition $P$, we have:
  \begin{align*}
    P \,\!& \leftrightarrow  P = \unit \quad
    && \mbox{(by propositional extensionality)}
    \\
        &\leftrightarrow  f(P)=f(\unit)
    && \mbox{(because $f$ is left-cancellable)}
    \\
        &\leftrightarrow f(P)=\Empty
    && \mbox{(by the assumption that $f(\unit)=\Empty$)}
    \\
        &\leftrightarrow \neg f(P)
          && \mbox{(by propositional extensionality).}
  \end{align*}
  (Note that if $\neg f(P)$, then $f(P) \leftrightarrow \Empty$, so $f(P)$ is a
  proposition and we can apply propositional extensionality to get
  $f(P)=\Empty$.)
  Hence $P$ is logically equivalent to the negation of the type
  $f(P)$, and therefore Lemma~\ref{dne:dne-negation} gives \DNE.
\end{proof}

\begin{corollary}
\label{thm:automorphism}
Assuming propositional extensionality, if there is an automorphism of the
universe that maps the unit type to the empty type, then \DNE\ holds.
\end{corollary}

Now let us further assume univalence and propositional truncations.
This implies function extensionality, so the difference between \DNE\ and \LEM\ disappears.
Furthermore, we can additionally generalize the result as follows.  Say
that a type $A$ is \emph{inhabited} if the unique map $A\to\unit$ is
surjective. This is equivalent to giving an element of the propositional
truncation~$\trunc{A}$.
\begin{lemma}\label{lem:prop-equivalent}
  Assuming univalence and propositional truncations, if $A$ is an
  inhabited type, then any proposition $P$ is logically equivalent to
  the identity type $(P \times A) = A$.
\end{lemma}
\begin{proof}
  If $P$ then $P \equiv \unit$, so $(P \times A) \equiv A$, and hence
  by univalence $(P \times A) = A$.  Conversely, assume $(P \times A)
  = A$.  Then $\trunc{P \times A} = \trunc{A} = \unit$ by univalence,
  as $A$ is inhabited.  So $\trunc{P}\times\trunc{A}=\unit$, and hence
  $P=\unit$.
\end{proof}
Using this, we can weaken the hypothesis of Lemma~\ref{left:cancellable} to the requirement that $f$
maps some inhabited type to the empty type, and get the same
conclusion, at the expense of
requiring univalence rather than just propositional extensionality:
\begin{lemma} \label{left:cancellable:bis}
  Assuming univalence and propositional truncations, if there is a
  left-cancellable map $f: \UU \to \UU$ with $f(A)=\Empty$ for some
  inhabited type $A$, then excluded middle holds.
\end{lemma}
\begin{proof}
  For an arbitrary proposition $P$, we have:
  \begin{align*}
    P \,\!& \leftrightarrow  (P \times A) = A \quad
    && \mbox{(by Lemma~\ref{lem:prop-equivalent})}
    \\
        &\leftrightarrow  f(P \times A)=f(A)
    && \mbox{(because $f$ is left-cancellable)}
    \\
        &\leftrightarrow f(P \times A)=\Empty
    && \mbox{(by the assumption that $f(A)=\Empty$)}
    \\
        &\leftrightarrow \neg f(P \times A)
          && \mbox{(by propositional extensionality).}
  \end{align*}
  Hence $P$ is logically equivalent to the negation of the type
  $f(P\times A)$, and therefore Lemma~\ref{dne:dne-negation} gives
  \DNE. But univalence gives function extensionality, and hence
  Lemma~\ref{lem:is:dne} gives \LEM.
\end{proof}
\begin{theorem}
\label{thm:automorphism:bis}
Assuming univalence and propositional truncations, if there is an
automorphism of the universe that maps some inhabited type to the
empty type, then excluded middle holds.
\end{theorem}

\begin{corollary}\label{cor:0:beers}
  Assuming univalence and propositional truncations, if there is an
  automorphism $g : \UU \to \UU$ of the universe with $g(\Empty) \ne
  \Empty$, then the double negation \[ \neg \neg \prod_{P:\UU}
  \isProp(P) \to P + \neg P \] of the law of excluded middle
  holds.
\end{corollary}
(Note that this is not the same as 
\[ \prod_{P:\UU}  \isProp(P) \to \neg \neg (P + \neg P ), \]
which is of course constructively valid without extra assumptions.)

\begin{proof}
  Let $f$ be the inverse of $g$. If $g(\Empty)$ then
  $\trunc{g(\Empty)}$, and because $f$ maps $g(\Empty)$ to
  $\Empty$, we conclude that excluded middle holds by
  Theorem~\ref{thm:automorphism:bis}.  But the assumption
  $g(\Empty) \ne \Empty$ is equivalent to $\neg \neg g(\Empty)$ by
  propositional extensionality, and so it implies the double negation of
  excluded middle.
\end{proof}

It is in general an open question for which $X$ the existence of an
automorphism $f:\UU \to \UU$ with $f(X)\neq X$ implies a non-provable
consequence of excluded middle~\cite{automorphisms:escardo:bet}. Not
even for $X=\unit$ do we know whether this is the case. However, the
following two cases for $X$ follow from the case $X=\Empty$ discussed
above:
\begin{corollary}\label{cor:lem:beers}
  Assuming univalence and propositional truncations, for universes
  $\UU:\VV$, if there is an automorphism $f:\VV \to \VV$ with
  $f(X) \ne X$ for $X=\LEM_\UU$ or $X=\neg \neg \LEM_\UU$, then
  $\neg \neg \LEM_\UU$ holds.
\end{corollary}
\begin{proof}
  Suppose that $\neg\LEM_\UU$, and hence $X=\Empty$.  By
  Corollary~\ref{cor:0:beers}, we obtain $\neg\neg\LEM_\VV$, which
  implies $\neg\neg \LEM_\UU$, contradicting the assumption.
\end{proof}

\typesorarxiv{
  \subparagraph*{Acknowledgements.}
}{
  \paragraph*{Acknowledgements.}
}

The first-named author would like to thank Uday Reddy for discussions
about parametricity.
We would also like to thank Jean-Philippe Bernardy for helpful comments, and Andrej Bauer for
discussions and questions.
The fact that the implication $\DNE \to \LEM$ requires $\Empty$-valued
function extensionality was spotted by one of the referees. Since
then, we have formalized the results to make sure we didn't miss
similar assumptions. All results of the paper, except
Section~\ref{sec:autom-that-do}, are formalized in
Agda~\cite{auke:development}. All results of
Sections~\ref{sec:autom-that-do} and~\ref{sec:excluded-middle-from}
are formalized in Coq~\cite{hottcoq} 
(in the files \verb+Spaces/BAut/Rigid.v+ and \verb+Spaces/Universe.v+).
(So some results have been formalized twice, and no numbered result has been left unformalized.)

\appendix

\bibliographystyle{plainurl}
\bibliography{parametricity-lem}

\end{document}